\documentclass[11pt, a4paper]{article}
\usepackage{amsmath,amsfonts,amsthm,amssymb,bm,tikz-cd}
\usepackage[mathscr]{eucal}
\usepackage[margin=1.1in]{geometry}
\usepackage[english]{babel}
\usepackage{multirow}
\usepackage{authblk}
\usepackage{tikz-cd} 
\usepackage{url}
\usepackage{color}
\usepackage{hyperref}
\newcommand\id{\mathrm{id}}
\newcommand\C{\mathscr C}

\newcommand\Set{\mathbf{Set}}
\newcommand\Grp{\mathbf{Grp}}

\newcommand\Top{\mathbf{Top}}

\newcommand\Vect{\mathbf{Vec}}
\newcommand\GCA{\mathcal{GCA}}
\newcommand\End{\mathrm{End}}
\newcommand\Mor{\mathrm{Mor}}
\newcommand\CA{\mathcal{CA}}

\newcommand\Hom{\mathrm{Hom}}
\newcommand\Obj{\mathrm{Obj}}

\newcommand\Dom{\mathrm{Dom}}
\newcommand\Cod{\mathrm{Cod}}

\theoremstyle{plain}

\newtheorem{corollary}{Corollary}

\newtheorem{theorem}{Theorem}

\theoremstyle{definition}
\newtheorem{definition}{Definition}
\newtheorem{example}{Example}
\newtheorem{remark}{Remark}

\begin{document}

\title{Categorical products of cellular automata}
\author[1]{Alonso Castillo-Ramirez\footnote{Email: alonso.castillor@academicos.udg.mx}}
\author[1]{Alejandro Vazquez-Aceves \footnote{Email: alejandro.vazquez5702@alumnos.udg.mx }}
\author[1]{Angel Zaldivar-Corichi \footnote{Email: luis.zaldivar@academicos.udg.mx }}
\affil[1]{Centro Universitario de Ciencias Exactas e Ingenier\'ias, Universidad de Guadalajara, M\'exico.}

\maketitle

\begin{abstract}
We study two categories of cellular automata. First, for any group $G$, we consider the category $\CA(G)$ whose objects are configuration spaces of the form $A^G$, where $A$ is a set, and whose morphisms are cellular automata of the form $\tau : A_1^G \to A_2^G$. We prove that the categorical product of two configuration spaces $A_1^G$ and $A_2^G$ in $\CA(G)$ is the configuration space $(A_1 \times A_2)^G$. Then, we consider the category of generalized cellular automata $\GCA$, whose objects are configuration spaces of the form $A^G$, where $A$ is a set and $G$ is a group, and whose morphisms are $\phi$-cellular automata of the form $\mathcal{T} : A_1^{G_1} \to A_2^{G_2}$, where $\phi : G_2 \to G_1$ is a group homomorphism. We prove that a categorical weak product of two configuration spaces $A_1^{G_1}$ and $A_2^{G_2}$ in $\GCA$ is the configuration space $(A_1 \times A_2)^{G_1 \ast G_2}$, where $G_1 \ast G_2$ is the free product of $G_1$ and $G_2$. The previous results allow us to naturally define the product of two cellular automata in $\CA(G)$ and the weak product of two generalized cellular automata in $\GCA$. 
\\

\end{abstract}

\section{Introduction}\label{intro}

Category theory is a powerful foundation of mathematics that provides a unifying framework to describe and analyze different structures and their relationships. Many constructions and concepts that appear in different context throughout mathematics, such as quotient spaces, direct products, adjoints, and duality, are unified with a precise definition in category theory. In particular, a \emph{categorical product}, or simply a \emph{product}, of two objects in a category is a generalization of the Cartesian product of two sets, the direct product of two groups, and the product topology of two topological spaces. This captures the idea of combining two objects to produce a new object, while preserving in some way its relationship with the original objects.  

In this paper, we study products in categories of cellular automata (CA). Recall that a \emph{cellular automaton} over a group $G$, known as the \emph{universe}, and a set $A$, known as the \emph{alphabet}, is a transformation $\tau : A^G \to A^G$ of the configuration space $A^G := \{ x : G \to A \}$ such that there exists a finite subset $S \subseteq G$ and a local function $\mu : A^S \to A$ that determines the behavior of $\tau$ (see Definition \ref{def-CA}). Traditionally, CA have been studied when $G = \mathbb{Z}^d$, but interest in the more general setting of an arbitrary group has recently grown (e.g., see \cite{CAG}). 

Before reviewing our work with more detail, we shall mention a few previous studies of cellular automata from a categorical point of view. In \cite{Siltar}, Silvio Capobianco and Tarmo Uustalu explore various properties of CA, such as the Curtis-Hedlund theorem and the reversibility principle, using category theory. In \cite{Cec2013}, Tullio Ceccherini-Silberstein and Michel Coornaert considered cellular automata $\tau : A^G \to B^G$ where $G$ is a fixed group, and $A$ and $B$ are objects in a concrete category. They prove analogous to various important results in the theory of CA, including theorems of surjunctivity and reversibility. Finally, in \cite{Ville2014}, Ville Salo and Ilkka T\"orm\"a examined various categories whose objects are subshifts of $A^{\mathbb{Z}}$ and whose morphisms are cellular automata, or block maps, between subshifts.  

Our work concerns two different categories of cellular automata. The first one is the the category $\CA(G)$ of all cellular automata over a fixed group $G$, whose objects are configuration spaces of the form $A^G$, where $A$ is any set, and whose morphisms are cellular automata $\tau : A_1^G \to A_2^G$ as defined in Definition \ref{def-CA}. The second one is the category $\GCA$ of \emph{generalized cellular automata}, whose objects are configuration spaces of the form $A^G$, where $A$ is a set and $G$ is a group, and whose morphisms are $\phi$-cellular automata $\mathcal{T} : A_1^{G_1} \to A_2^{G_2}$, where $\phi : G_2 \to G_1$ is a group homomorphism, as defined in \cite{Vaz2022}. 

We show the categorical product of two configuration spaces $A_1^G$ and $A_2^G$ in the category $\CA(G)$ is the configuration space $(A_1 \times A_2)^G$ together with the projections 
\[ \pi_i^G : (A_1 \times A_2)^G \to A_i^G, \quad \pi_i^G := \pi_i \circ x, \forall x \in (A_1 \times A_2)^G, \]
where $\pi_i : A_1 \times A_2 \to A_i$ are the natural projections of the Cartesian product of sets. This allows us to define the categorical product for a pair of cellular automata $\tau_1 : A_1^G \to B_1^G$ and $\tau_2 : A_2^G \to B_2^G$ to be the cellular automaton $\tau_1 \times \tau_2 : (A_1 \times A_2)^G \to (B_1 \times B_2)^G$ such that
\[ (\tau_1 \times \tau_2)(x)(g) = (\tau_1(\pi_{A_1}^G(x)(g)), \tau_2(\pi_{A_2}^G(x)(g)),  \quad \forall x \in (A_1 \times A_2)^G, g \in G; \] 
this coincides which what is usually considered to be the product of two CA (see \cite[Ex. 1.17]{ExCAG}).

On the other hand, we show that  a \emph{weak product} of two configuration spaces $A_1^{G_1}$ and $A_2^{G_2}$ in the category $\GCA$ is the configuration space $(A_1 \times A_2)^{G_1 \ast G_2}$, where $G_1 \ast G_2$ denotes the free product of $G_1$ and $G_2$, together with projections induced by $\pi_i$ and the canonical embeddings $\iota_i : G_i \to G_1 \ast G_2$ (see Theorem \ref{th-weak}). The definition of a weak product in a category is analogous to a product, except that the morphism that satisfies the universal property is only required to exist but not to be unique. The reason behind why this weak product in $\GCA$ fails to be a product is because generalized CA may fail to have the \emph{unique homomorphisms property} as introduced in \cite{Further}; this means that there are generalized CA $\mathcal{T} : A_1^{G_1} \to A_2^{G_2}$ that are $\phi$- and $\psi$-cellular automata for $\phi \neq \psi$. 

The structure of this paper is as follows. In Section 2 we review some basic notions of category theory, including the product and coproduct, as well as some examples. In Section 3 we show that the category of cellular automata has products, while in Section 4 we show review some basic theory of generalized cellular automata and show that their category has a weak product.

\section{Category theory}

In this section we shall review some basic notions of category theory; for a broader treatment see \cite{Handbook,Sau1970,RomanCat,IntroCat}.

A \emph{category} $\C$ consists on a class of \emph{objects}, denoted by $\Obj(\C)$, and a class of \emph{morphisms}, or \emph{arrows}, denoted by $\Mor(\C)$. Each morphism $f \in \Mor(\C)$ has a source object $\Dom(f) \in \Obj(\C)$, which we call the \emph{domain} of $f$, and an \emph{arrival} object $\Cod(f)\in \Obj(\C)$, which we call the \emph{codomain} of $f$. All the information that a morphism carries is condensed in the notation 
\[ f : \Dom(f)\to\Cod(f). \]
The class of morphisms with domain $A$ and codomain $B$ is denoted by $\Hom_\C(A,B)$. Part of the essence of category theory is that objects are treated as if they are not necessarily sets, and morphisms are treated as if they are not necessarily functions. 

The morphisms in a category must satisfy some additional axioms. For each pair of morphisms $f, g\in\Mor(\C)$, such that $\Cod(f)=\Dom(g)$, there exists a morphism $g\circ f : \Dom(f) \to \Cod(g)$, called a \emph{composition} of $g$ with $f$, which satisfies the following properties: 
\begin{enumerate}
\item For any objects $A,B,C,D\in\Obj(\C)$, and morphisms $f:A\to B,g:B\to C$ and $h:C\to D$, we have
\[ h\circ(g\circ f) = (h\circ g)\circ f. \]

\item For every object $A\in\Obj(\C)$, there exists a morphism $\id_A\in\Hom_\C( A,A)$ that satisfies the following: for every $B\in\Obj(\C)$ and $g\in\Hom_\C(A,B)$, $h\in\Hom_\C(B,A)$, 
\[ g\circ \id_A=g\text{ and }\id_A\circ h=h. \]
\end{enumerate}

To simplify the notation, when referring to an object or a morphism of a category $\C$, we will write $A\in\C$ or $f\in\C$, instead of $A\in\Obj(\C)$ or $f\in\Mor(\C)$. Moreover, when there is no doubt about the category in question, the collection $\Hom_\C(A,B)$ will be denoted simply by $\Hom(A,B)$. Morphisms in $\Hom(A,A)$ are called \emph{endomorphisms}, so we will denote such a collection as $\End(A)$.

\begin{example}\label{Ex1}
We introduce some examples of categories.
\begin{enumerate}
\item The class of sets, together with the class of all functions between sets, form the \emph{category of sets}, which we denote as $\Set$.

\item The class of all groups, together with the class of all group homomorphisms, form the \emph{category of groups}, which we denote as $\Grp$.

\item The class of all topological spaces, together with the class of all continuous functions, form the \emph{category of topological spaces}, which we denote as $\Top$.

\item Given a field $\mathbb{K}$, the class of vector spaces over $\mathbb{K}$, together with the class of linear transformations, form the \emph{category of vector spaces} over $\mathbb{K}$, which we denote by $\Vect_{\mathbb{K}}$.

\item Let $M$ be a monoid, that is, a set with an associative binary operation and with identity $1\in M$. We can view $M$ as a category with one object, i.e., $\Obj(M)=\{\bullet\}$,
and each element in $M$ is a morphism in the category, i.e., $\Mor(M)=M$, with the monoid operation as composition of morphisms. In this category the endomorphisms of the only object are all the morphisms of the category, i.e., $M=\End(\bullet)$.

\item Let $(P,\leq)$ be a partially ordered set, we can view $P$ as a category, where $\Obj(P) =P$, and given $a,b\in P$, there is a unique morphism from $a$ to $b$ if and only if $a\leq b$. By the properties of the order relation it is clear that all the category axioms are satisfied.
\end{enumerate}
\end{example}

In the first four examples, the objects of the categories are sets with additional structure, and the morphisms are functions that preserve these structures. These categories are called \emph{concrete categories}. On the other hand, the last two examples seek to point out that morphisms do not necessarily have to be functions between sets.

Now we introduce the definition of the product in an arbitrary category.

\begin{definition}[product]\label{def-prod}
Let $\C$ be a category. A \emph{product} of two objects $A$ and $B$ in $\C$ is an object $A\times B\in\C$ together with morphisms $\pi_1:A\times B\to A$ and $\pi_2:A\times B\to B$ satisfying the following universal property: for any object $X$ and any morphisms $f:X\to A$ and $g:X\to B$, there exists a unique morphism $h:X\to A\times B$ such that $\pi_1\circ h=f$ and $\pi_2\circ h=g$; this is equivalent of saying that the diagram of Figure \ref{dia-prod} commutes.
\begin{figure}[h!]
\centering
\begin{tikzcd}
& X \arrow[ld, "f"'] \arrow[rd, "g"] \arrow[d, "h", dashed] & \\
A & A\times B \arrow[l, "\pi_1"] \arrow[r, "\pi_2"'] & B
\end{tikzcd}
\caption{Definition of a product to two objects in a category.}
\label{dia-prod}
\end{figure}
\end{definition}

It turns out that if a product of two objects exists, then it is unique up to isomorphism. The morphisms $\pi_1:A\times B\to A$ and $\pi_2:A\times B\to B$ are called the \emph{projections} to $A$ and $B$, respectively. We say that the category $\C$ itself has \emph{finite products} if any two objects in $\C$ have a product.

A \emph{weak product} of two objects $A\times B\in\C$ is defined analogously as the product in Definition \ref{def-prod}, except that the morphism $h:X\to A\times B$ is not required to be unique (see \cite[Sec. X.2]{Sau1970}).

\begin{example}\label{Examples of products}\text{\\}
We introduce a few examples of products in categories.
\begin{enumerate}

\item In the concrete categories described above, the product of two objects $A$ and $B$ is the Cartesian product of sets $A\times B$ equipped with the corresponding structure so that the projections $\pi_1:A\times B\to A$ and $\pi_2:A\times B\to B$, defined as $\pi_1(a,b)=a$ and $\pi_2(a,b)=b$, are morphisms in the corresponding categories ($\Set,\Grp,\Top,\Vect_{\mathbb{K}}$, etc.). However, there exist concrete categories where the product is not the Cartesian product together with the projections described above. When it does coincide, we say that the concrete category has \emph{concrete products} (see \cite[Ex. 10.55]{concrete}).

\item Consider a set $A$ and the partially ordered set $(2^A,\subseteq)$, where $2^A$ is the power set of $A$. Considering $2^A$ as a category as in Example \ref{Ex1}.6., the product of two objects $X,Y\in 2^A$ is the set $X\cap Y$ together with the unique arrows $X\cap Y\to X$ and $X\cap Y\to Y$ satisfying that $X\cap Y\subseteq X$ and $X\cap Y\subseteq Y$.

\item Consider as a category the partially ordered set $(\mathbb{N},|)$, where $\mathbb{N}$ is the set of natural numbers and $|$ is the divisibility relation. The product of two objects $m,n\in\mathbb{N}$ is the greatest common divisor of $m$ and $n$ together with the unique morphisms given by being divisors of $m$ and $n$.
\end{enumerate}
\end{example}

When a category $\C$ has finite products we may define the \emph{product of two morphisms} of $\C$.

\begin{theorem}\label{prod-mor}
Let $\C$ be a category with finite products, and let $f_1 : A_1 \to B_1$ and $f_2 : A_2 \to B_2$ be two morphisms. Then there exists a unique morphism $f_1 \times f_2 : A_1 \times A_2 \to B_1 \times B_2$ such that
\[ \pi_{B_1} \circ (f_1 \times f_2)  = f_1 \circ \pi_{A_1} \text{ and } \pi_{B_1} \circ (f_1 \times f_2) = f_2 \circ \pi_{A_2}, \]  
where $\pi_{A_i} : A_1 \times A_2 \to A_i$ and $\pi_{B_i} : B_1 \times B_2 \to B_i$ are the projection morphisms; this is equivalent of saying that the diagram of Figure \ref{dia-prodmor} commutes.
\begin{figure}[h!]
\centering
\begin{tikzcd}
A_1 \arrow[d, "f_1"']  & A_1 \times A_2 \arrow[l, "\pi_{A_1}"'] \arrow[r, "\pi_{A_2}"] \arrow[d, "f_1 \times f_2", dashed] & A_2 \arrow[d, "f_2"'] \\
B_1 & B_1\times B_2 \arrow[l, "\pi_{B_1}"] \arrow[r, "\pi_{B_2}"'] & B_2
\end{tikzcd}
\caption{Product to two morphisms in a category.}
\label{dia-prodmor}
\end{figure}
\end{theorem}
\begin{proof}
This follows from the universal property of the product with $X=A_1 \times A_2$, $f = f_1\circ \pi_{A_1}$ and $g = f_2 \circ \pi_{A_2}$ (c.f \cite[p. 30]{RomanCat}).
\end{proof}

\begin{example}
In the category $\Set$, the product of two functions $f_1 : A_1 \to B_1$ and $f_2 : A_2 \to B_2$ is the function $f_1 \times f_2 : A_1 \times A_2 \to B_1 \times B_2$ defined by
\[ (f_1 \times f_2)(a_1, a_2) := (f_1(a_1) , f_2(a_2)), \quad \forall (a_1,a_2 ) \in A_1 \times A_2. \]
\end{example}

The dual notion of a product in a category, called a \emph{coproduct}, is defined by reversing the direction of the morphisms in the definition of a product. Despite this simple change in the definition, products and coproducts may behave dramatically different in a category.

\begin{definition}[coproduct]
Let $\C$ be a category, and $A,B\in\C$ be objects of the category. A \emph{coproduct} of objects $A$ and $B$ in $\C$ is an object $A+ B\in\C$ together with morphisms $\iota_1:A\to A+ B$ and $\iota_2:B\to A+ B$ satisfying the following universal property: for any object $X$ and morphisms $f:A\to X$ and $g:B\to X$ there exists a unique morphism $h:X\to A+ B$ such that $h\circ\iota_1=f$ and $h\circ\iota_2=g$; this is equivalent of saying that the diagram of Figure \ref{dia-coprod} commutes.
\begin{figure}[h!]
\centering
\begin{tikzcd}
& X & \\
A \arrow[ru, "f"] \arrow[r, "\iota_1"'] & A+ B \arrow[u, "h", dashed] & B \arrow[l, "\iota_2"] \arrow[lu, "g"']
\end{tikzcd}
\caption{Definition of the coproduct to two objects in a category.}
\label{dia-coprod}
\end{figure}
\end{definition}

\begin{example}\label{Examples of coproducts}\text{\\}
We introduce some examples of coproducts in categories. 
\begin{enumerate}
\item In the category of sets $\Set$, the coproduct of $A$ and $B$ corresponds to the disjoint union $A+ B=(A\times\{A\})\cup(B\times\{B\})$ together with the inclusions $\iota_1:A\to A+ B$ and $\iota_2:B\to A+ B$, defined as
$\iota_1(a)=(a,A)$ and $\iota_2(b)=(b,B)$.

\item In the category of groups $\Grp$, the coproduct of $G$ and $H$ corresponds to the free product of groups $G\ast H$ together with the canonical embeddings $\iota_1:G\to G\ast H$
and $\iota_2:H\to G\ast H$.

\item Consider as a category the partially ordered set $(\mathbb{N},|)$. The coproduct of two objects $m,n\in\mathbb{N}$ is the least common multiple of $m$ and $n$ together with the unique morphism given by being a multiple of $m$ and $n$.
\end{enumerate}
\end{example}


\section{Product of cellular automata}

The category of cellular automata over a group $G$, denoted by $\CA(G)$, consists of objects that are \emph{configuration spaces} over $G$; this is, sets of $A^G:=\{x:G\to A \}$ where $A$ is a set, usually called an \emph{alphabet}. The morphisms of $\CA(G)$ are \emph{cellular automata}, as given in the following definition, which is a slightly more general than Definition 1.4.1 in \cite{CAG}, as it involves two arbitrary alphabets $A$ and $B$ (see \cite[Sec. 1.1.4]{ExCAG}).

\begin{definition}\label{def-CA}
Let $A^G$ and $B^G$ be two configuration spaces. We say that a function $\tau:A^G\to B^G$ is a \emph{cellular automaton} if there exists a finite subset $S\subseteq G$ (called a \emph{memory set}) and a function $\mu:A^S\to B$ (called a \emph{local function}) such that:
\[ \tau(x)(g) = \mu((g \cdot x)|_S), \quad \forall x \in A^G, g \in G,  \]
where $\cdot$ is the \emph{shift action} of $G$ on $A^G$ given by
\[ (g\cdot x)(h)=x(hg), \quad \forall x\in A^G,  g,h\in G. \]
and $x|_S$ is the restriction of $x \in A^G$ to $S$. 
\end{definition}

\begin{remark}
Alternatively as introduced in Definition \ref{def-CA}, we may consider the shift action of $G$ on $A^G$ to be
\[ (g\star x)(h)=x(g^{-1}h), \quad \forall x\in A^G,  g,h\in G. \]
However, we may check that $\cdot$ and $\star$ are \emph{equivalent actions} in a formal sense: there exists a bijection $\beta : A^G \to A^G$ given by $\beta(x)(g) := x(g^{-1})$, for all $x \in A^G$, $g \in G$, such that 
\[ \beta(g \cdot x) = g \star \beta(x), \quad \forall x \in A^G, g \in G. \]  
\end{remark}

Since the composition of two cellular automata is a cellular automaton by \cite[Prop. 1.4.9]{CAG} and the identity function of $A^G$ is a cellular automaton, then $\CA(G)$ is a category.

\begin{example}\label{ex-ind}
For any function $f:A\to B$ between sets we may define a cellular automaton $f^G : A^G\to B^G$ by
\[ f^G(x) := f \circ x , \quad \forall x\in A^G. \]
It is clear that $f^G$ has memory set $\{e \}$, where $e \in G$ is the identity element of the group, and local function $f : A^{\{e\}} \to A$.
\end{example}

The next result shows that the category $\CA(G)$ has concrete finite products. 

\begin{theorem}\label{th-prod}
Let $G$ be a group, and let $A_2^G$ and $A_2^G$ be two configuration spaces in $\CA(G)$. The product of $A_1^G$ and $A_2^G$ in the category $\CA(G)$
is the configuration space $(A_1 \times A_2)^G$ together with the cellular automata $\pi_i^G: (A_1\times A_2)^G\to A_i^G$ induced, as in Example \ref{ex-ind}, by the usual projections $\pi_i : A_1 \times A_2 \to A_i$ in the category of sets. 
\end{theorem}
\begin{proof}
Let $B^G\in\CA(G)$ be a configuration space, and consider two cellular automata $\tau_1:B^G\to A_1^G$ and $\tau_2:B^G\to A_2^G$. Let $S_1, S_2\subseteq G$ be memory sets, and let $\mu_1: B^{S_1} \to A_1$ and $\mu_2: B^{S_2} \to A_2$ be local functions for $\tau_1$ and $\tau_2$, respectively. Define a finite set $S := S_1\cup S_2\subseteq G$ and local functions
\[  \hat{\mu}_1 : B^S\to A_1 \text{ and } \hat{\mu}_2 :B^S \to A_2,  \]
by $\hat{\mu}_1(x) := \mu_1(x \vert_{S_1})$ and $\hat{\mu}_2(x):= \mu_2(x \vert_{S_2})$, for all $x \in B^S$. The universal property of the product in the category of sets tells us that there exists a unique function 
\begin{equation}\label{eq-proj}
 \mu:B^S\to A_1 \times A_2 \text{ such that } \pi_i \circ \mu= \hat{\mu}_i, \quad i \in \{1,2\}. 
\end{equation}
 Now let $\tau:B^G\to(A_1\times A_2)^G$ be the cellular automaton defined by the local function $\mu$. We shall check that the diagram given in Figure \ref{dia-prodCA} commutes.
\begin{figure}[h!]
\centering
\begin{tikzcd}
 & B^G \arrow[ld, "\tau_1"'] \arrow[rd, "\tau_2"] \arrow[d, "\tau"] & \\
 A_1^G & (A_1\times A_2)^G \arrow[l, "\pi_1^G"] \arrow[r, "\pi_2^G"'] & A_2^G
 \end{tikzcd}
 \caption{Product of two configuration spaces in $\CA(G)$.}
\label{dia-prodCA}
\end{figure}
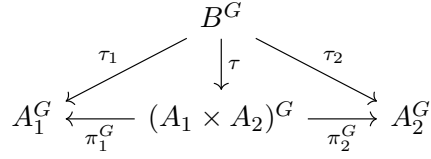
Indeed, for any $x\in B^G$ and $g\in G$, we have
\begin{eqnarray*}
(\pi_i^G\circ\tau)(x)(g) &=& \pi_i ( \tau(x) (g) ) \\
&= &\pi_i(\mu((g\cdot x)|_S))\\
&=& \hat{\mu}_i((g\cdot x)|_S)  \quad \quad \quad \text{(by eq. (\ref{eq-proj}))} \\ 
&=& \mu_i((g\cdot x)|_{S_i})\\
&=& \tau_i(x)(g).
\end{eqnarray*}
Therefore, $\pi_i^G\circ\tau = \tau_i$, for $i \in \{1,2\}$. 

Finally, to prove that $\tau$ is unique, suppose there exists another cellular automaton $\sigma:B^G\to(A_1\times A_2)^G$ such that $\tau_i=\pi_i^G\circ\sigma$. Then, for all $x\in B^G$ and $g\in G$ we have:
\[ \pi_i(\sigma(x)(g)) = \pi_i^G(\sigma(x))(g)=\pi_i^G(\tau(x))(g)=\pi_i(\tau(x)(g))\] 
for $i \in \{1,2 \}$. Hence $\sigma(x)(g)=\tau(x)(g)$ for all $x\in B^G$ and $g\in G$,
which implies that $\sigma=\tau$.
\end{proof}

\begin{corollary}
For any cellular automata $\tau_1 : A_1^G \to B_1^G$ and $\tau_2 : A_2^G \to B_2^G$ in $\CA(G)$ there is a unique cellular automaton $\tau_1 \times \tau_2 : (A_1 \times A_2)^G \to (B_1 \times B_2)^G$ such that
\[ \pi_{B_1}^G \circ (\tau_1 \times \tau_2)  = \tau_1 \circ \pi_{A_1}^G \text{ and } \pi_{B_2}^G \circ (\tau_1 \times \tau_2) = \tau_2 \circ \pi_{A_2}^G,  \]  
This is equivalent to,
\[(\tau_1 \times \tau_2)(x)(g) = (\tau_1(\pi_{A_1} \circ x) (g), \tau_2(\pi_{A_2} \circ x)(g)),  \quad \forall x \in (A_1 \times A_2)^G, g \in G. \] 
\end{corollary}


\section{Weak product of generalized cellular automata}

The main idea in the definition of \emph{generalized cellular automata} is to consider CA between configuration spaces over different groups. The following is a slightly generalized version of the definition introduced in \cite{Vaz2022}, as it also involves two different alphabets.

\begin{definition}
Let $\phi : H \to G$ be a group homomorphism, and let $A$ and $B$ be sets. A \emph{$\phi$-cellular automaton} is a function $\mathcal{T} : A^G \to B^H$ such that there is a finite subset $S \subseteq G$ and a local function $\mu : A^S \to B$ satisfying
\[ \mathcal{T}(x)(h) = \mu( (\phi(h) \cdot x) \vert_S), \quad \forall x \in A^G, h \in H.  \]
\end{definition}

We say that $\mathcal{T} : A^G \to B^H$ is a \emph{generalized cellular automaton} if it is a $\phi$-cellular automaton for some group homomorphism $\phi : H \to G$. Note that when $G = H$, an $\id$-cellular automaton is the same as a cellular automaton over $G$. 

\begin{example}
Consider the group of integers $\mathbb{Z}$ and the direct product of $\mathbb{Z}^2$. Define a homomorphism $\phi : \mathbb{Z}^2 \to \mathbb{Z}$ by $\phi(n,m) := n+m$, for all $(n,m) \in \mathbb{Z}^2$. Let $S:= \{-1,0,1\} \subseteq \mathbb{Z}$. Then, any local function $\mu : A^S \to B$ defines a $\phi$-cellular automaton $\mathcal{T} : A^{\mathbb{Z}} \to B^{\mathbb{Z}^2}$ by
\[ \tau(x)(n,m) = \mu(x(n+m-1), x(n+m), x(n+m +1)), \quad \forall x \in A^{\mathbb{Z}}, (n,m ) \in \mathbb{Z}^2. \] 
\end{example}

As shown in \cite[Lemma 4]{Further}), for every $\phi$-cellular automaton $\mathcal{T} : A^G \to B^H$ there exists a unique cellular automaton $\tau : A^G \to B^G$ such that 
\begin{equation}\label{factor}
 \mathcal{T} = \phi^*_B \circ \tau,
\end{equation}
where $\phi^*_B : B^G \to B^H$ is the $\phi$-cellular automaton defined by
\[ \phi^*_B(x) = x \circ \phi, \quad \forall x \in B^G. \]
It is easy to check that for every pair of homomorphisms $\phi : H \to G$ and $\psi : K \to H$, we have
\begin{equation}\label{eq-phi}
 (\phi \circ \psi)_B^* = \psi^*_B \circ \phi^*_B.   
\end{equation}

It follows from \cite[Theorem 2]{Vaz2022} that the composition of a $\phi$-cellular automaton $\mathcal{T} : A^G \to B^H$ with a $\psi$-cellular automaton $\mathcal{S} : B^H \to C^K$ is a $(\psi \circ \phi)$-cellular automaton. Hence, we shall consider the category $\GCA$ whose objects are configuration spaces of the form $A^G$, where $A$ is a set and $G$ is a group, and whose morphisms are $\phi$-cellular automata $\mathcal{T} : A_1^{G_1} \to A_2^{G_2}$, where $\phi : G_2 \to G_1$ is a group homomorphism. 

For two groups $G_1$ and $G_2$, let $G_1 \ast G_2$ be their free product. The cannonical embeddings $\iota_i : G_i \to G_1 \ast G_2$ induce $\iota_i$-cellular automata 
\[ (\iota_i)_C^* : C^{G_1 \ast G_2} \to C^{G_i}, \]
for any alphabet $C$.

\begin{theorem}\label{th-weak}
A weak product of two configuration spaces $A_1^{G_1}$ and $A_2^{G_2}$ in $\GCA$ is the configuration space $(A_1 \times A_2)^{G_1 \ast G_2}$ together with the projections 
\[ \gamma_i : (A_1 \times A_2)^{G_1 \ast G_2} \to A_i^{G_i} \quad \text{ given by } \quad \gamma_i := (\iota_i)_{A_i}^*  \circ \pi_{A_i}^{G_1 \ast G_2}, \] 
where $\pi_{A_i}^{G_1 \ast G_2}$ are the cellular automata projections considered in Theorem \ref{th-prod}.
\end{theorem}

\begin{figure}[h!] 
\centering
\begin{tikzcd}
 & & B^H \arrow[swap, ld, "\tau_1"] \arrow[rd, "\tau_2"] \arrow[d, "\tau"] & &   \\[1em]
& A_1^H \arrow[bend right=30, ld, "\phi_{A_1}^*"] \arrow[d, "(\phi + \psi)_{A_1}^*"]    & (A_1\times A_2)^H \arrow[l, "\pi_{A_1}^H"] \arrow[r, "\pi_{A_2}^H"']\arrow[d, "(\phi +\psi)_{A_1 \times A_2}^*"] & A_2^H \arrow[bend right=-30, rd, "\psi_{A_2}^*"]  \arrow[d, "(\phi + \psi)_{A_2}^*"]  &  \\[1em]
 A^{G_1}_1 & A_{1}^{G_1 \ast G_2} \arrow[l, "(\iota_1)_{A_1}^*"]  & (A_1 \times A_2)^{G_1 \ast G_2} \arrow[l, "\pi_{A_1}^{G_1 \ast G_2}"] \arrow[swap, r, "\pi_{A_2}^{G_1 \ast G_2}"] & A_2^{G_1 \ast G_2} \arrow[swap, r, "(\iota_2)_{A_2}^*"]   & A_2^{G_2}
 \end{tikzcd}
 \caption{Proof of the existence of a weak product in $\GCA$.}
\label{dia-proof}
 \end{figure}
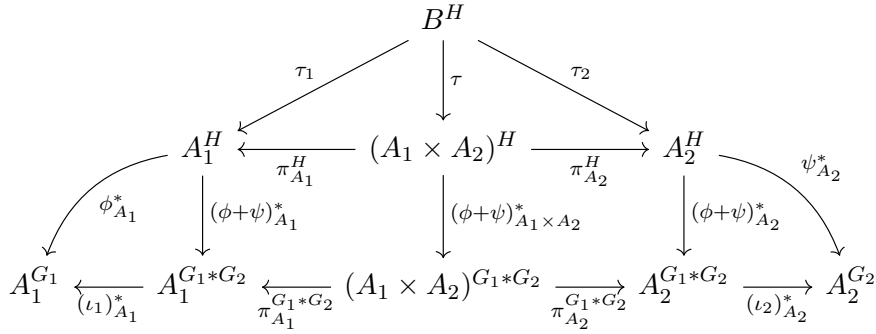
 
\begin{proof}
Let $\phi^*_{A_1} \circ \tau_1 : B^H \to A_1^{G_1}$ and $\psi^*_{A_2} \circ \tau_2 : B^H \to A_2^{G_2}$ be two generalized CA, where $\phi : G_1 \to H$ and $\psi : G_2 \to H$ are group homomorphisms and $\tau_i : B^H \to A_i^{G_i}$ are cellular automata. By the universal property of the coproduct in the category of groups, there exists a unique group homomorphism $\phi + \psi : G_1 \ast G_2 \to H$ such that the following diagram commutes:
\[\begin{tikzcd}
& H & \\[1em]
G_1 \arrow[ru, "\phi"] \arrow[r, "\iota_1"'] & G_1 \ast G_2 \arrow[u, "\phi + \psi"] & G_2 \arrow[l, "\iota_2"] \arrow[lu, "\psi"']
\end{tikzcd}\]
Equation (\ref{eq-phi}) implies that the following diagram commutes for every alphabet $C$:
\begin{equation}\label{eq-dia}
\begin{tikzcd}
& C^H \arrow[bend right=30, ld, "\phi^*_C"'] \arrow[bend right=-30, rd, "\psi^*_C"] \arrow[d, "(\phi + \psi)^*_C"] & \\[1em]
C^{G_1} &  C^{G_1 \ast G_2}  \arrow[l, "(\iota_1)^*_C"] \arrow[r, "(\iota_2)^*_C"'] & C^{G_2}
\end{tikzcd}
\end{equation}

Our goal is to show that the diagram of Figure \ref{dia-proof} commutes, where the cellular automaton $\tau : B^H \to (A_1 \times A_2)^H$ is given by the universal property of the product of two cellular automata in $\CA(H)$ (Theorem \ref{th-prod}). Clearly, the top triangle of diagram \ref{dia-proof} commutes because of Theorem \ref{th-prod}. Note that diagram (\ref{eq-dia}) with $C=A_1$ and $C=A_2$ gives us the commutativity of the leftmost and rightmost lower triangles of diagram \ref{dia-proof}, respectively. Finally, we will show that the left and right squares of diagram \ref{dia-proof} commute. For example, for the right square, note that for all $x \in (A_1 \times A_2)^H$, we have
\[  (\phi + \psi)^*_{A_2} \circ \pi_{A_2}^H(x) = (\pi_{A_2} \circ x) \circ (\phi + \psi) =  \pi_{A_2} \circ (x \circ (\phi + \psi)) = \pi_{A_2}^{G_1 \ast G_2} \circ (\phi + \psi)^*_{A_1 \times A_2} (x).  \] 
Similarly, the left square commutes.  
\end{proof}

The problem with showing that the weak product given in Theorem \ref{th-weak} is unique, and hence a product, is that the factorization (\ref{factor}) of generalized CA is not unique in general. We say that $\mathcal{T}$ has the \emph{unique homomorphism property} (UHP) if $\mathcal{T} = \phi^*_B \circ \tau = \psi_B^* \circ \tau$ implies $\phi = \psi$. It was shown in \cite[Corollary 1]{Further}, that if $G$ is a toersion-free abelian group, then every non-constant $\mathcal{T}$ has the UHP; this result was generalized  in \cite{SaloBlog} by removing the hypothesis of $G$ being abelian. However, note that constant generalized CA never have the UHP: if $\mathcal{T} : A^G \to B^H$ is constant, then $\mathcal{T} = \phi^*_B \circ \tau = \psi_B^* \circ \tau$ for every pair of homomorphisms $\phi$ and $\psi$. 

The coproduct in the category of groups gives us a coproduct of group homomorphisms: for any group homomorphisms $\phi : H_1 \to G_1$ and $\psi : H_2 \to G_2$ there exists a unique group homomorphisms $\phi \ast \psi : H_1 \ast H_2 \to G_1 \ast G_2$ such that the dual of the diagram in Figure \ref{dia-prodmor} commutes. Our last result shows the existence of a weak product of morphisms in $\GCA$.

\begin{corollary}
For any $\phi$-cellular automaton $\mathcal{T} : A_1^{G_1} \to B_1^{H_1}$ and a $\psi$-cellular automaton $\mathcal{S} : A_2^{G_2} \to B_2^{H_2}$ there exists $(\phi \ast \psi)$-cellular automaton
\[ \mathcal{T} \times \mathcal{S} : (A_1 \times A_2)^{G_1 \ast G_2} \to (B_1 \times B_2)^{H_1 \ast H_2},  \]
such that the diagram of Figure \ref{dia-last} commutes.
\begin{figure}[h!]
\centering
\begin{tikzcd}
A_1^{G_1} \arrow[d, "\mathcal{T}" ] & A_1^{G_1 \ast G_2} \arrow[l, "(\iota_1)_{A_1}^*"]  & (A_1 \times A_2)^{G_1 \ast G_2} \arrow[l, "\pi_{A_1}^{G_1 \ast G_2}"'] \arrow[r, "\pi_{A_2}^{G_1 \ast G_2}"] \arrow[d, "\mathcal{T} \times \mathcal{S}", dashed] & A_2^{G_1 \ast G_2} \arrow[r, "(\iota_2)_{A_2}^*"] & A_2^{G_2} \arrow[d, "\mathcal{S}" ] \\
B_1^{H_1} & B_1^{H_1 \ast H_2}  \arrow[l, "(\iota_1)_{B_1}^*" ]  & (B_1\times B_2)^{H_1 \ast H_2} \arrow[l, "\pi_{B_1}^{H_1 \ast H_2}"] \arrow[r, "\pi_{B_2}^{H_1 \ast H_2}"'] & B_2^{H_1 \ast H_2} \arrow[r, "(\iota_2)_{B_2}^*" ] & B_2^{H_2}
\end{tikzcd}
\caption{Weak product to two generalized cellular automata.}
\label{dia-last}
\end{figure}
\end{corollary}



\section*{Acknowledgments}

The second author was supported by a CONAHCYT \emph{Becas nacionales para estudios de posgrado}. The third author was supported by a CONAHCYT grant No. CBF-2023-2024-2630.


\bibliography{mibib}{}
\bibliographystyle{plain}

\end{document}